\documentclass[letterpaper, 12 pt]{article}  
 
\usepackage[left=2.2cm, right=2.2cm, top=2cm, bottom=2cm]{geometry}                     

\usepackage{lipsum}

\newcommand\blfootnote[1]{%
  \begingroup
  \renewcommand\thefootnote{}\footnote{#1}%
  \addtocounter{footnote}{-1}%
  \endgroup
}

\usepackage{algcompatible}
\usepackage[linesnumbered,ruled,vlined]{algorithm2e}

\usepackage{balance} 
\usepackage{xcolor,calc}

\usepackage{graphics}
\usepackage{epstopdf}
\usepackage{graphicx}
\usepackage{amsmath} 
\usepackage{amssymb}  
\usepackage{amsthm}
\usepackage[noadjust]{cite}
\usepackage{color}
\usepackage{enumerate}
\usepackage{subfigure}
\usepackage{multirow}
\usepackage{mathtools}
\usepackage{booktabs}
\usepackage{epstopdf}
\usepackage[normalem]{ulem}
\usepackage{hyperref}
\usepackage{url}
\usepackage{bm}
\usepackage{caption}

\usepackage{booktabs}

\newtheorem{remark}{Remark}
\newtheorem{proposition}{Proposition}

\makeatletter
\newsavebox\myboxA
\newsavebox\myboxB
\newlength\mylenA

\newcommand*\xoverline[2][0.75]{%
    \sbox{\myboxA}{$\m@th#2$}%
    \setbox\myboxB\null
    \ht\myboxB=\ht\myboxA%
    \dp\myboxB=\dp\myboxA%
    \wd\myboxB=#1\wd\myboxA
    \sbox\myboxB{$\m@th\overline{\copy\myboxB}$}
    \setlength\mylenA{\the\wd\myboxA}
    \addtolength\mylenA{-\the\wd\myboxB}%
    \ifdim\wd\myboxB<\wd\myboxA%
       \rlap{\hskip 0.5\mylenA\usebox\myboxB}{\usebox\myboxA}%
    \else
        \hskip -0.5\mylenA\rlap{\usebox\myboxA}{\hskip 0.5\mylenA\usebox\myboxB}%
    \fi}
\makeatother

\usepackage{tabularx}
\usepackage{booktabs}

\newtheorem{definition}{Definition}
\newtheorem{assumption}{Assumption}

\makeatletter
\let\NAT@parse\undefined
\makeatother

\begin{document}

\title{Learning to Satisfy Unknown Constraints in Iterative MPC}

\author{Monimoy Bujarbaruah$^{\dagger}$, Charlott Vallon$^{\dagger}$, and Francesco Borrelli\blfootnote{$\dagger$ authors contributed equally to this work. Emails:\{monimoyb, charlott, fborrelli\}@berkeley.edu.}
}

\maketitle

\begin{abstract}
We propose a control design method for
linear time-invariant systems that iteratively learns to satisfy unknown polyhedral state constraints. 
At each iteration of a repetitive task, the method constructs an estimate of the unknown environment constraints using collected closed-loop trajectory data. This estimated constraint set is improved iteratively upon collection of additional data. 
An MPC controller is then designed to robustly satisfy the estimated constraint set. 
This paper presents the details of the proposed approach, and provides robust and probabilistic guarantees of constraint satisfaction as a function of the number of executed task iterations.
We demonstrate the efficacy of the proposed framework in a detailed numerical example. 
\end{abstract}

\section{Introduction}\label{sec:intro}
Data-driven decision making and control has garnered significant attention in recent times \cite{tanaskovic2017data, recht2018tour, rosolia2018data, hewing2019learning, pourbabaee2020robust}. As such approaches are increasingly being deployed in automated systems \cite{liniger2015optimization, todorovHand, FelixDrone, losey2019learning}, the satisfaction of safety requirements is of utmost importance. \emph{Safety} is often represented as containment of system states (or outputs) within a pre-defined constraint set over all possible time evolutions of the considered system. Such constraint sets define the safe \emph{environment} for the system, in which the system is allowed to evolve during execution of a control task. Various control methods exist for ensuring system safety during a control task execution \cite{GaliACC, ames2016control, borrelli2017predictive}. 

In the additional presence of uncertainty in the system model, data-driven methods have been used to quantify and bound the uncertainty in order to ensure system safety either robustly \cite{fastrack, bujarArxivAdap, kohlerNTrack, sumeetCon}, or with high probability \cite{berkenkamp2017safe, hewing2017cautious, soloperto2018learning, kollerLearning}. The majority of existing methods assume that the environment constraints are known to the control designer. If the environment constraints are unknown, data-driven methods can be used to first learn the unknown constraints \cite{armestoCL, clearnShah, chou2018learning} and then design safe controllers using one of the previous methods. However, these approaches assume a perfectly known system model, not subject to any disturbances. 
The literature on safe, data-driven controller design in the presence of uncertainties in \textit{both} the system model and the constraint set is rather limited. 
In particular, such methods typically are unable to quantify the probability of the system failing to satisfy the true environment constraints.

In this paper we propose an algorithm to design a safe controller for an uncertain system while learning polyhedral state constraints. This discussion is motivated by systems performing tasks in dynamic, unsupervised environments. For example, can a robotic manipulator playing the kendama game\footnote{\href{https://www.youtube.com/watch?v=YhDabRISH04}{https://www.youtube.com/watch?v=YhDabRISH04}} in an obstacle-free environment adapt its strategy to the unannounced addition of a nearby wall? As a first approach, we consider a linear time-invariant system with known system matrices, subject to an additive disturbance, performing an iterative task. 
The environment constraints of the task are assumed polyhedral, characterized by a set of hyperplanes, some of which are unknown to the control designer. We assume that violations of the unknown constraints can be directly measured or observed from closed-loop state trajectories.

Our algorithm iteratively constructs estimates of the unknown constraints using collected system trajectories. These estimates are then used to design a  robust MPC controller \cite{Goulart2006, bujarArxivAdap} for safely achieving the control task despite the uncertainty. 



The main contributions of this paper are as follows:
\begin{itemize}
    \item Given a user-specified upper bound $\epsilon$ on the probability of violating the true constraint set $\mathbb{Z}$ within any $j^\mathrm{th}$ task iteration, we construct constraint estimates $\hat{\mathbb{Z}}^j$ from previously collected closed-loop task data, using  convex hull operations (for $\epsilon=0$) or a Support Vector Machine (SVM) classifier (for $\epsilon \in (0,1)$). We then design an MPC controller to robustly satisfy $\hat{\mathbb{Z}}^j$ along the $j^\mathrm{th}$ iteration, for all possible additive disturbance values.  
    \item When $\hat{\mathbb{Z}}^j$ is formed with the SVM classification approach (for $\epsilon \in (0,1)$), we provide an explicit number of successful task iterations to obtain before the estimated set $\hat{\mathbb{Z}}^j$ is deemed safe with respect to $\epsilon$. Here, ``successful task iterations" refers to closed-loop trajectories satisfying the unknown constraints $\mathbb{Z}$.
     
    \item When $\hat{\mathbb{Z}}^j$ is formed using the convex hull approach (for~$\epsilon = 0$), we show how to design a robust MPC that provides satisfaction of the true constraints $\mathbb{Z}$ at all future iterations $k\geq j$. 
\end{itemize}
The algorithm further highlights an exploration-exploitation trade-off well known in bandits literature \cite{gupta2018active, gupta2020unified, gupta2020correlated, gupta2021best, gupta2021multi, gupta2022structured, cho2020bandit}. The remainder of the paper is organized as follows: In Section~\ref{sec:probS} we formulate the robust optimization problem to be solved in each iteration, and define the inherent system dynamics along with state and input constraints. In Section~\ref{sec:iter_mpc} the MPC optimization problem is presented along with a definition of Iteration Failure under unknown (or partially known) state constraints. Section~\ref{sec:icl} delineates the control design requirements while finding approximations of the unknown constraints and consequently presents the associated algorithms. Finally, we present detailed numerical simulations corroborating our results in Section~\ref{sec:simul}.  

\section{Problem Setup}\label{sec:probS}
We consider linear time-invariant systems of the form:
\begin{equation}\label{eq:unc_system}
	x_{t+1} = Ax_t + Bu_t + w_t,
\end{equation}
where $x_t\in \mathbb{R}^{n}$ is the state at time $t$, $u_t\in\mathbb{R}^{m}$ is the input, and $A$ and $B$ are known system matrices. At each time step $t$, the system is affected by an independently and identically distributed (i.i.d.) \ random disturbance $w_t$ with a known polytopic support $\mathbb{W} \subset \mathbb{R}^{n}$. We define ${H}_x \in \mathbb{R}^{s \times n}$, ${h}_x \in \mathbb{R}^s$, $H_u \in \mathbb{R}^{o \times m}$, and $h_u \in \mathbb{R}^o$, and formulate the state and input constraints imposed by the task environment for all time steps $t \geq 0$ as:
\begin{align}\label{eq:constraints_nominal}
	\mathbb{Z} & := \{(x,u): {H}_x x 
	\leq {h}_x,~ {H}_u u \leq h_u \}.
\end{align}
Throughout the paper, we assume that system \eqref{eq:unc_system} performs the same task repeatedly, with each task execution referred to as an \emph{iteration}. 
Our goal is to design a controller that, at each iteration $j$, aims to solve the following finite horizon robust optimal control problem:
\begin{equation}\label{eq:generalized_InfOCP}
	\begin{array}{clll}
		\hspace{0cm} V^{j,\star}(x_S) = \\ [1ex]
		\displaystyle\min_{u_0^{j},u_1^{j}(\cdot),\ldots, u^j_{T-1}(\cdot)} & \displaystyle\sum\limits_{t=0}^{T-1} \ell \left( \bar{x}_t^{j}, u_t^{j}\left(\bar{x}_t^{j}\right) \right) \\[1ex]
		\text{s.t.,}  & x_{t+1}^{j} = Ax_t^{j} + Bu_t^{j}(x_t^{j}) + w_t^{j},\\
		& \bar{x}_{t+1}^{j} = A \bar{x}_t^{j} + Bu_t^{j}(\bar{x}_t^{j}),
		\\[1ex]
		& H_x x_t^{j} \leq h_x,~\forall w_t^{j} \in \mathbb W,\\[1ex]
		& H_u u_t^{j} \leq h_u,~\forall w_t^{j} \in \mathbb W,\\[1ex]
		&  x_0^{j} = x_S,\ t=0,1,\ldots,(T-1),
	\end{array}
\end{equation}
where $x_t^{j}$, $u_t^{j}$ and $w_t^{j}$ denote the realized system state, control input and disturbance at time $t$ of the $j^{\mathrm{th}}$ iteration respectively. The pair $(\bar{x}_t^{j}, u_t^j(\bar{x}_t^j))$ denotes the disturbance-free nominal state and corresponding nominal input. The optimal control problem \eqref{eq:generalized_InfOCP} minimizes the nominal cost over a time horizon of length $T \gg 0$ at any $j^\mathrm{th}$ iteration with $j \in \{1, 2, \dots\}$. 
The state and input constraints must be robustly satisfied for all uncertain realizations.
The optimal control problem \eqref{eq:generalized_InfOCP} consists of finding $[u_0^{j},u_1^{j}(\cdot),u_2^{j}(\cdot),\ldots]$, where $u_t^{j}: \mathbb{R}^{n}\ni x_t^{j} \mapsto u_t^{j} = u_t^{j}(x_t^{j})\in\mathbb{R}^{m}$ are state feedback policies. 

In this work we consider constraints of the form:
\begin{align}\nonumber
    H_x & = \begin{bmatrix} H^\mathrm{b}_x \\ H^\mathrm{ub}_x\end{bmatrix}, 
    h_x = \begin{bmatrix} h^\mathrm{b}_x \\ h^\mathrm{ub}_x\end{bmatrix},
\end{align}
where the superscripts $\{\mathrm{b, ub}\}$ denote the known and unknown parts of the constraints, respectively. 
That is to say, we consider a scenario in which we only know a subset of the system's environment constraint set.
At the beginning of the $j^\mathrm{th}$ task iteration we construct approximations of $H_x$ and $h_x$, denoted as $\hat{H}^j_x$ and $\hat{h}^j_x$, respectively, using closed-loop trajectories of the system from previous task iterations. The estimated constraints form a safe set estimate $\hat{\mathbb{Z}}^j$:
\begin{align}\label{eq:estim_con}
    \hat{\mathbb{Z}}^j := \{(x, u) : \hat{H}^j_x x \leq \hat{h}^j_x, H_u u \leq h_u\}.
\end{align}
These estimates are refined iteratively using new data as the system continues to perform the task, and are used to solve an estimate of \eqref{eq:generalized_InfOCP}. 
The construction of the safe set estimates is detailed in Section~\ref{sec:icl}.

\section{Iterative MPC Problem}\label{sec:iter_mpc}
For computational tractability when considering task duration $T \gg 0$, we try to approximate a solution to the optimal control problem \eqref{eq:generalized_InfOCP} by solving a simpler constrained optimal control problem with prediction horizon $N \ll T$ in a receding horizon fashion. 
\subsection{Problem Definition}
Since  the true constraint set $\mathbb{Z}$ is not completely known, we use our estimate $\hat{\mathbb{Z}}^j$ built from data and formulate the robust optimal control problem as: 
\begin{equation} \label{eq:MPC_R_fin}
	\begin{aligned}
	  V_{t \rightarrow t+N}&^{\mathrm{MPC},j}(x^j_t, \hat{\mathbb{Z}}^j,\hat{\mathcal{X}}^j_N)  :=	\\
	& \min_{U^j_t(\cdot)} ~~ \sum_{k=t}^{t+N-1} \ell(\bar{x}^j_{k|t}, v^j_{k|t}) + Q(\bar{x}^j_{t+N|t})\\
		& ~~~\text{s.t.,}~~~    x^j_{k+1|t} = Ax^j_{k|t} + Bu^j_{k|t} + w^j_{k|t},\\
		& ~~~~~~~~~~~\bar{x}^j_{k+1|t} = A\bar{x}^j_{k|t} + Bv^j_{k|t},\\
		&~~~~~~~~~~~u^j_{k|t} = \sum \limits_{l=t}^{k-1}M^j_{k,l|t} w^j_{l|t}  + v^j_{k|t},\\
		&~~~~~~~~~~~ \hat{H}^j_x x^j_{k|t} \leq \hat{h}^j_x,\\
		&~~~~~~~~~~~ H_u u^j_{k|t} \leq h_u,\\
		&~~~~~~~~~~~x^j_{t|t} = \bar{x}^j_{t|t},\\
	    &~~~~~~~~~~~ {x}^j_{t+N|t} \in \hat{\mathcal{X}}_N^j,\\
	    & ~~~~~~~~~~~ \forall w^j_{k|t} \in {\mathbb{W}},\forall k = \{t,\ldots,t+N-1\},
	\end{aligned}
\end{equation}
where in the $j^\mathrm{th}$ iteration, $x^j_t$ is the measured state at time $t$, $x^j_{k|t}$ is the predicted state at time $k$, obtained by applying predicted input policies $[u^j_{t|t},\dots,u^j_{k-1|t}]$ to system~\eqref{eq:unc_system}.
We denote the disturbance-free nominal state and corresponding input as $\{\bar{x}^j_{k|t}, v^j_{k|t}\}$ with $v^j_{k|t} = u^j_{k|t}(\bar{x}^j_{k|t})$. The MPC controller minimizes the cost over the predicted nominal trajectory $\Big \{ \{\bar{x}^j_{k|t}, v^j_{k|t}\}_{k=t}^{t+N-1}, \bar{x}^j_{t+N|t} \Big \}$, which is comprised of a positive definite stage cost $\ell(\cdot, \cdot)$ and terminal cost $Q(\cdot)$.  
We note that the above formulation uses affine disturbance feedback parameterization \cite{Goulart2006} of input policies.  
We use state feedback $u^j_t = Kx^j_t$ with $(A+BK)$ being stable to construct a terminal set $\hat{\mathcal{X}}^j_N = \{x \in \mathbb{R}^n: \hat{Y}^j x \leq \hat{z}^j,~\hat{Y}^j \in \mathbb{R}^{r^j \times n},~\hat{z}^j \in \mathbb{R}^{r^j}\}$, which is the $(T-N)$ step robust reachable set \cite[Chapter~10]{borrelli2017predictive} to the set of state constraints in \eqref{eq:estim_con}. 
Specifically, this set has the properties:
\begin{equation}\label{eq:term_set_DF}
    \begin{aligned}
    &\hat{\mathcal{X}}^j_N \subseteq \{x~|~(x,Kx) \in \hat{\mathbb{Z}}^j\},\\
    &\hat{H}^j_x((A+BK)^ix + \sum \limits_{\tilde{i}=0}^{i-1} (A+BK)^{i-\tilde{i}-1}w_{\tilde{i}}) \leq \hat{h}^j_x,\\
    &H_u(K ( (A+BK)^ix + \sum \limits_{\tilde{i}=0}^{i-1} (A+BK)^{i-\tilde{i}-1}w_{\tilde{i}} )) \leq h_u,\\
    &\forall x\in \hat{\mathcal{X}}^j_N,~\forall w_{i} \in {\mathbb{W}},~\forall i=1,2,\dots,(T-N).
    \end{aligned}
\end{equation}
After solving \eqref{eq:MPC_R_fin} at time step $t$ of the $j^\mathrm{th}$ iteration, we apply
\begin{equation}\label{eq:inputCL_DF}
	u^j_t = v^{j,\star}_{t|t}
\end{equation}
to system \eqref{eq:unc_system}. 
We then resolve the problem \eqref{eq:MPC_R_fin} again at the next $(t+1)$-th time step, yielding a receding horizon strategy. 

We note that computing the set $\hat{\mathcal{X}}^j_N$ in \eqref{eq:term_set_DF} at each iteration can be computationally  expensive. In such cases one can opt for data-driven methods such as \cite{rosolia2017learning, kimPstuff} or simple approximation methods such as \cite{girard2007approximation, kurzhanski2000ellipsoidal} to construct these terminal sets.

\begin{assumption}[Well-Posedness of Task]\label{assump:well_posed}
We assume that given an initial task state $x_S$, the optimization problem \eqref{eq:MPC_R_fin} is  feasible at all times $0 \leq t \leq T-1$ for the true constraint set $\hat{\mathbb{Z}}^j = \mathbb{Z}$ as defined in \eqref{eq:constraints_nominal}, for all iterations $j \in \{1,2,\dots\}$. We further assume that $\mathbf{0}_{n \times 1} \in \mathbb{Z}$. 
\end{assumption}


\subsection{Successful Task Iterations}
At each iteration, the true constraint set $\mathbb{Z}$ is unknown and being estimated with $\hat{\mathbb{Z}}^j$ built from data. Depending on how $\hat{\mathbb{Z}}^j$ is constructed, robust satisfaction of the true constraints \eqref{eq:constraints_nominal} during an iteration may not be guaranteed. It is thus possible that \eqref{eq:constraints_nominal} becomes infeasible at some point while solving \eqref{eq:MPC_R_fin} during $0 \leq t \leq T-1$ along any $j^\mathrm{th}$ iteration. We formalize this with the following definition:
\begin{definition}[Successful Iteration]\label{def:succ_iter}
A Successful $j^\mathrm{th}$ Iteration is defined as the event  
\begin{align}\label{eq:iter_failure}
    [\mathrm{SI}]^j : H_x x_t^j \leq h_x,~\forall t \in [0,T].
\end{align}
That is, an iteration is successful if there are no state constraint violations during $0 \leq t \leq T$. Otherwise, the iteration is deemed failed; that is, an Iteration Failure event is implicitly defined as $[\mathrm{IF}]^j = ([\mathrm{SI}]^j )^\textnormal{c}$, where $([\cdot])^\textnormal{c}$ denotes the complement of an event. 
\end{definition}
The probability of a Successful Iteration $[\mathrm{SI}]^j$ is a function of the sets $\hat{\mathbb{Z}}^j$ since $x_t^j$ is the closed-loop trajectory obtained when applying the feedback controller
\eqref{eq:MPC_R_fin}-\eqref{eq:inputCL_DF}.

\subsection{Control Design Objectives}
Our aim is not only to keep the probability of $[\mathrm{IF}]^j$ low along each iteration, but also to maintain satisfactory controller performance in terms of cost during successful iterations. Let the closed-loop cost of a successful iteration $j$ under observed disturbance samples $w^{j}$ be denoted by
\begin{align*}
    \hat{\mathcal{V}}^j(x_S, w^{j}) = \sum \limits_{t=0}^{T-1} \ell({x}^j_t,v^{j,\star}_{t|t}),
\end{align*}
where notation $w^j$ denotes $[w^j_0, w^j_1, \dots, w^j_{T-1}]$. We use the average closed-loop cost $\mathbb{E}[\hat{\mathcal{V}}^j(x_S, w^{j})]$ to quantify controller performance. Specifically, our goal is to lower the iteration \emph{performance loss}, defined as
\begin{align} \label{eq:cl_loop_diff}
    [\mathrm{PL}]^j = \mathbb{E} [\hat{\mathcal{V}}^j(x_S, w^{j})] - \mathbb{E} [{\mathcal{V}}^{\star}(x_S, w^{j})],
\end{align}
where $\mathbb{E} [{\mathcal{V}}^{\star}(x_S, w^{j})]$ denotes the average closed-loop cost of an iteration if $\mathbb{Z}$ had been known, i.e. if $\hat{\mathbb{Z}}^j = \mathbb{Z}$ for all $j \in \{1,2,\dots\}$. To formalize this joint focus on obtaining a low probability of Iteration Failures while maintaining satisfactory controller performance, we summarize our control design objectives as:
\begin{enumerate}[(C1)]
    \item Design a closed-loop MPC control law \eqref{eq:inputCL_DF} which ensures that the system \eqref{eq:unc_system} maintains a user-specified upper bound on the probability of Iteration Failure $[\mathrm{IF}]^j$ (\ref{eq:iter_failure}), \emph{for all} iterations $j \in \{1,2,\dots\}$.
    
    \item Minimize $[\mathrm{PL}]^j$ (as defined in \eqref{eq:cl_loop_diff}) at \emph{each} iteration $j \in \{1,2,\dots\}$, while satisfying (C1).
\end{enumerate}
However, as we start the control task from scratch without assuming the initial availability of a large number of trajectory data samples, and it is difficult in general to obtain statistical properties of estimated constraint sets $\hat{\mathbb{Z}}^j$, methods such as \cite{calafiore2006scenario, zhang:margellos:goulart:lygeros:13, bujarbaruah2019learning } cannot be used to satisfy (C1)-(C2) directly. We therefore relax the above two specifications and formulate two control design specifications (D1) and (D2) in the next section. 

\section{Iterative Constraint Learning}\label{sec:icl}
We consider the following design specifications: 
\begin{enumerate}[(D1)]
    \item Design a closed-loop MPC control law \eqref{eq:inputCL_DF} which ensures that the system \eqref{eq:unc_system} maintains a user-specified upper bound $\epsilon$ on the probability of Iteration Failure, \emph{after} some iteration $j \in \{1,2,\dots\}$.
    
    \item Minimize $[\mathrm{PL}]^j$ (as defined in \eqref{eq:cl_loop_diff}) \emph{after} some iteration $j \in \{1,2,\dots\}$.
\end{enumerate}
 We wish to find the smallest index $\bar{j}$, such that (D1) and (D2) are satisfied for all $j \geq \bar{j}$. 
The design specifications (D1)-(D2) indicate that the approach to construct estimated state constraint sets proposed in this paper, is our best possible attempt to satisfy (C1)-(C2), given the information available at each iteration $j$.

\begin{assumption}[Feasibility Classification]\label{ass:feasflag}
Given a system state trajectory, we assume that a classifier is available to check the feasibility of each point in the trajectory based on whether it satisfies the true state constraints in \eqref{eq:constraints_nominal}. This classifier returns a corresponding sequence of feasibility flags. 
\end{assumption}

\begin{assumption}[Simulator] \label{ass:sim}
We assume that each iteration is run until completion at time $T$, and that state constraint satisfaction as described in Assumption~\ref{ass:feasflag} is checked only at the end of the simulation. 
\end{assumption}

We note that Assumption~\ref{ass:sim} could be relaxed in several ways. 
For example, constraint satisfaction could be checked in real-time and the simulations stopped if violations occur. 
One could also run physical experiments and check the feasibility of \eqref{eq:constraints_nominal} in real-time, by observing if the physical experiment fails.
Some constraint violations may be hard to evaluate during physical experiments, but this discussion goes beyond the scope of this paper.

\subsection{Constructing Constraint Estimates $\hat{\mathbb{Z}}^j$}\label{ssec:hatZ_constr}
We show how the estimated constraint sets $\hat{\mathbb{Z}}^j$ are constructed in order to satisfy the design specifications (D1) and (D2). This process depends on the user-specified upper bound $\epsilon$ on the probability of Iteration Failure. To satisfy (D1) we search for the smallest $\bar{j}$, such that
\begin{align}\label{eq:math_def_IF}
    \mathbb{P}([\mathrm{IF}]^j) \leq \epsilon,
\end{align}
for all $j \geq \bar{j}$, where $\epsilon \in (0,1)$ is the bound on the probability of Iteration Failure. At the start of the first iteration, $j=1$, we use only the known information about the imposed constraints:
\begin{align}\label{eq:zhat1}
    \hat{\mathbb{Z}}^1 := \{(x, u) : {H}^\mathrm{b}_x x \leq {h}^\mathrm{b}_x, H_u u \leq h_u\}.
\end{align}
Next, consider any $j \in \{1,2,\dots\}$. Let the closed-loop realized states collected until the end of the $j^\mathrm{th}$ iteration be
\begin{align} \label{eq:traj_jm1}
    \mathbf{x}^{1 : j} = [x^{1 : j}_0, x^{1 : j}_1, \dots, x^{1 : j}_T],
\end{align}
where $x^{1:j}_i \in \mathbb{R}^{n \times j}$ is a matrix containing all states corresponding to time step $i$ from the first $j$ iterations.
Let $f^j(x) : \mathbb{R}^n \mapsto \mathbb{R}$ denote a curve that separates the points in \eqref{eq:traj_jm1} according to
whether they satisfy all true state constraints in \eqref{eq:constraints_nominal}, such that $f^j(\mathbf{0}_{n \times 1}) \leq 0$. 
Based on Assumption~\ref{ass:feasflag}, such a binary classification curve can be obtained with supervised learning techniques. In this paper we use a kernelized Support Vector Machine algorithm \cite[Chapter~12]{friedman2001elements}.

Let a polyhedral inner approximation\footnote{Approximation techniques are elaborated in Section~\ref{sec:simul}.} of the intersection of $f^j(x) \leq 0$ and the known state constraints in $\hat{\mathbb{Z}}^1$ be given by:
\begin{align}\label{eq:pol_svm}
\hat{\mathcal{P}}^{j+1}_\mathrm{svm} &= \{x: \hat{H}^{j+1}_{x,\mathrm{svm}} x \leq \hat{h}^{j+1}_{x,\mathrm{svm}} \} \\
    & =  \{x:f^j(x) \leq 0\} \cap \{x: H^\mathrm{b}_x x \leq h^\mathrm{b}_x\}. \nonumber
\end{align}
We then use \eqref{eq:pol_svm} to form the constraint set estimates for the following iteration:
\begin{align}\label{eq:zhat_svm}
      \hat{\mathbb{Z}}_\mathrm{svm}^{j+1} := \{(x, u) : \hat{H}^{j+1}_{x, \mathrm{svm}} x \leq \hat{h}^{j+1}_{x, \mathrm{svm}}, H_u u \leq h_u\},
\end{align}
setting $\hat{\mathbb{Z}}^{j+1} = \hat{\mathbb{Z}}_\mathrm{svm}^{j+1}$ in our robust optimization problem \eqref{eq:MPC_R_fin} for $j \in \{1,2,\dots\}$.
In other words, at each iteration $j>1$, the estimated state constraints in $\hat{\mathbb{Z}}^{j}$ are formed out of the SVM classification boundary learned from all previous state trajectories, intersected with the known state constraints.

\begin{remark}\label{rem:scale_svmset}
In case the set $\hat{\mathbb{Z}}_\mathrm{svm}^{j+1}$ in \eqref{eq:zhat_svm} yields either infeasibility of \eqref{eq:MPC_R_fin} or an empty terminal set $\hat{\mathcal{X}}^{j+1}_N$ for any iteration $j \in \{1,2,\dots\}$, the set of estimated state constraints can be scaled appropriately until feasibility of \eqref{eq:MPC_R_fin} is obtained. Such scaling is not further analyzed in the remaining sections of this paper. 
\end{remark}

Since the estimated constraint sets \eqref{eq:zhat_svm} are not necessarily inner approximations of the true unknown constraints \eqref{eq:constraints_nominal}, the closed-loop state trajectories in future iterations may result in Iteration Failures with a nonzero probability. In the following proposition we quantify the probability of an Iteration Failure, given a $\hat{\mathbb{Z}}^{\bar{j}}$, for some $\bar{j} \in \{1,2,\dots\}$.

\begin{proposition}\label{prop:prob_con}
Consider $\hat{\mathbb{Z}}_\mathrm{svm}^1 = \hat{\mathbb{Z}}^1$ from \eqref{eq:zhat1} for $\bar{j} = 1$  or a constraint estimate set $\hat{\mathbb{Z}}_\mathrm{svm}^{\bar{j}}$ from \eqref{eq:zhat_svm} formed using trajectories up to iteration $\bar{j}-1$ for $\bar{j} >1$. 
Let this set $\hat{\mathbb{Z}}_\mathrm{svm}^{\bar{j}}$ be used as the constraint estimate set for the next $N_\mathrm{it}$ task iterations, beginning with iteration $\bar{j}$.
If for a chosen $\epsilon \in (0,1)$ and $0< \beta \ll 1$, Successful Iterations are obtained for the next $N_\mathrm{it} \geq \frac{\ln 1/\beta}{\ln 1/(1-\epsilon) }$ iterations, then 
$\mathbb{P}([\mathrm{IF}]^{j}) \leq \epsilon$ with confidence at least $1-\beta$ \emph{for all} subsequent task iterations $j\geq \bar{j}$ using $\hat{\mathbb{Z}}^{{j}}= \hat{\mathbb{Z}}_\mathrm{svm}^{\bar{j}}$. 
\end{proposition}
\begin{proof}
See Appendix.
\end{proof}
Proposition~\ref{prop:prob_con} requires that the polytope $\hat{\mathcal{P}}^{j+1}_\mathrm{svm}$ for $j \in \{1,2,\dots\}$ is updated \emph{only} if new violation points for constraints \eqref{eq:constraints_nominal} are seen at the end of an iteration $j$. This update strategy is highlighted in Algorithm~\ref{alg1}. 
 If no violations are seen for $N_\mathrm{it}$ successive iterations, a probabilistic safety certificate is provided 
and Algorithm~\ref{alg1} is terminated. 

\subsection{Safety vs Performance Trade-Off}\label{ssec:cvx_sec}
Proposition~\ref{prop:prob_con} proves that constructing estimated constraint sets as per \eqref{eq:zhat_svm}, can result in satisfaction of \eqref{eq:math_def_IF} for some $\epsilon \in (0,1)$. 
However, for certain applications, violations of constraints \eqref{eq:constraints_nominal} may be too expensive to allow for a nonzero probability of failure, and we instead require $\epsilon = 0$.
In such cases, we can utilize the closed-loop system trajectories for obtaining \textit{guaranteed} inner approximations of \eqref{eq:constraints_nominal}, so that 
$\mathbb{P}([\mathrm{IF}]^j) = 0$ for all future iterations $j \geq \bar{j}$, for some $\bar{j}$ to be determined.

Recalling \eqref{eq:traj_jm1}, let the closed-loop realized states collected until the end of the $j^
\mathrm{th}$ iteration be denoted as
\begin{align}\label{eq:traj_clcvx}
    \mathbf{x}^{1:j} = [x^{1:j}_0, x^{1:j}_1, \dots, x^{1:j}_T],
\end{align}
and let $\hat{\mathbf{x}}^{j}$ 
denote the collection of states from \eqref{eq:traj_clcvx} which satisfy all true state constraints in \eqref{eq:constraints_nominal}. Then an inner approximation the of state constraints in  \eqref{eq:constraints_nominal} is provided by the polyhedron:
\begin{align}\label{eq:conv_hull_pol}
\hat{\mathcal{P}}_\mathrm{cvx}^{j+1} &= \{x: \hat{H}^{j+1}_{x,\mathrm{cvx}} x \leq \hat{h}^{j+1}_{x,\mathrm{cvx}}\} \\
& = \mathrm{conv}([\mathbf{0}_{n \times 1}, \hat{\mathbf{x}}^j]) \nonumber,
\end{align}
where $\mathrm{conv}(\cdot)$ denotes the convex hull operator. We can now define
\begin{align}\label{eq:z_update_conv}
     \hat{\mathbb{Z}}_\mathrm{cvx}^{j+1} := \{(x, u) : \hat{H}^{j+1}_{x,\mathrm{cvx}} x \leq \hat{h}^{j+1}_{x,\mathrm{cvx}}, H_u u \leq h_u\},
\end{align}
and use $\hat{\mathbb{Z}}^j = \hat{\mathbb{Z}}_\mathrm{cvx}^j$ for $j \in \{2,3,\dots\}$ in \eqref{eq:MPC_R_fin} as a robust alternative to \eqref{eq:zhat_svm}. 

\begin{proposition}\label{prop:zhat_cvx}
If $\hat{\mathbb{Z}}_\mathrm{cvx}^{\bar{j}}$ \eqref{eq:z_update_conv} yields feasibility of \eqref{eq:MPC_R_fin} for some $\bar{j} \in \{2,3,\dots\}$, then
\begin{align*}
   \hat{\mathbb{Z}}_\mathrm{cvx}^{\bar{j}} \subseteq \mathbb{Z},
\end{align*}
and $\hat{\mathbb{Z}}_\mathrm{cvx}^{{j}} = \hat{\mathbb{Z}}_\mathrm{cvx}^{\bar{j}}$ for all $j \geq \bar{j}$. 
\end{proposition}

\begin{proof}
Let the closed-loop realized states collected until the end of the $(\bar{j}-1)^\mathrm{th}$ iteration be
\begin{align}\label{eq:traj_clcvxBarj}
    \mathbf{x}^{1:\bar{j}-1} = [x^{1:\bar{j}-1}_0, x^{1:\bar{j}-1}_1, \dots, x^{1:\bar{j}-1}_T],
\end{align}
and let $\hat{\mathbf{x}}^{\bar{j}-1}$ be the collection of all trajectory points in \eqref{eq:traj_clcvxBarj} that satisfy the state constraints in \eqref{eq:constraints_nominal}. Following \eqref{eq:conv_hull_pol} we form $\hat{\mathcal{P}}_\mathrm{cvx}^{\bar{j}} = \{x: \hat{H}^{\bar{j}}_{x,\mathrm{cvx}} x \leq \hat{h}^{\bar{j}}_{x,\mathrm{cvx}}\}$ as,
\begin{align*}
    \hat{\mathcal{P}}_\mathrm{cvx}^{\bar{j}}=\mathrm{conv}([\mathbf{0}_{n \times 1}, \hat{\mathbf{x}}^{\bar{j}-1}]).
\end{align*}
By the convexity of the true unknown state constraints \eqref{eq:constraints_nominal} and Assumption~\ref{assump:well_posed}, we have that $\hat{\mathcal{P}}^{\bar{j}}_\mathrm{cvx} \subseteq \{x: H_x x \leq h_x\}$. This implies $\hat{\mathbb{Z}}^{\bar{j}} \subseteq \mathbb{Z}$. 

Furthermore, since \eqref{eq:MPC_R_fin} is feasible at time $t=0$ in iteration $\bar{j}$, \eqref{eq:MPC_R_fin} remains feasible with system \eqref{eq:unc_system} in closed-loop with the MPC controller \eqref{eq:cl_loop_diff} at all future times $t\leq (T-1)$.\footnote{This recursive feasibility property is stated without proof. Interested readers can look into the standard detailed proofs in \cite[Chapter~12]{borrelli2017predictive}.} It follows that $x^{\bar{j}}_t \in \hat{\mathcal{P}}^{\bar{j}}_\mathrm{cvx}$ for all $0 \leq t \leq T$, which implies $\hat{\mathcal{P}}^{\bar{j}+1}_\mathrm{cvx} = \hat{\mathcal{P}}^{\bar{j}}_\mathrm{cvx}$ from \eqref{eq:conv_hull_pol}. Extending this argument, we can similarly prove $\hat{\mathcal{P}}^{j}_\mathrm{cvx} = \hat{\mathcal{P}}^{\bar{j}}_\mathrm{cvx}$ for all $ j >\bar{j}$, which implies $\hat{\mathbb{Z}}_\mathrm{cvx}^{{j}} = \hat{\mathbb{Z}}_\mathrm{cvx}^{\bar{j}}$ for all $j >\bar{j}$. This completes the proof.     
\end{proof}
Proposition~\ref{prop:zhat_cvx} implies that if we find a $\bar{j}$ for which \eqref{eq:z_update_conv} yields feasibility of \eqref{eq:MPC_R_fin}, then the probability of Iteration Failure at iteration $j$ is exactly $0$ for all $j \geq \bar{j}$. Moreover, Proposition~\ref{prop:zhat_cvx} suggests that after the $\bar{j}^\mathrm{th}$ iteration, the constraint estimation update \eqref{eq:z_update_conv} can be terminated.

The update strategy \eqref{eq:z_update_conv} strictly ensures that $\hat{\mathbb{Z}}_\mathrm{cvx}^{j} \subseteq \mathbb{Z}$ for all $j \in \{2,3,\dots\}$, which is not necessarily true for sets obtained using the SVM method \eqref{eq:zhat_svm}. 
However, choosing this robust constraint estimation can increase the performance loss \eqref{eq:cl_loop_diff} over successful iterations after $j \geq \bar{j}$. This is the \emph{safety vs. performance trade-off}, which the user can manage with an appropriate choice of $\epsilon$. \textcolor{black}{Given any $\epsilon \in (0,1)$ or $\epsilon =0$, the chosen strategy lowers performance loss while satisfying (D1). Thus we satisfy (D2) with (D1).}
\begin{remark}
Following Remark~\ref{rem:scale_svmset}, if the optimization problem \eqref{eq:MPC_R_fin} is infeasible or the terminal set $\hat{\mathcal{X}}_N^j$ constructed in \eqref{eq:term_set_DF} using the estimate $\hat{\mathbb{Z}}_\mathrm{cvx}^j$ is empty, one can switch to constraint estimates \eqref{eq:zhat_svm} and collect additional trajectory data, since 
$\hat{\mathcal{P}}_\mathrm{cvx}^{j_1} \subseteq \hat{\mathcal{P}}_\mathrm{cvx}^{j_2}$, for any $2 \leq j_1 < j_2$. 
\end{remark}




\subsection{The RMPC-ICL Algorithm}
We present our Robust MPC with Iterative Constraint Learning (RMPC-ICL) algorithm, which uses the estimated constraint sets $\hat{\mathbb{Z}}^j$ from Section~\ref{ssec:hatZ_constr} or Section~\ref{ssec:cvx_sec} while solving \eqref{eq:MPC_R_fin} in an iterative fashion. The algorithm terminates upon finding the smallest $\bar{j}$ such that \eqref{eq:math_def_IF} is satisfied. 
\begin{algorithm}[h!]
    \caption{
    RMPC-ICL Algorithm 
    }
    \label{alg1}
    \begin{algorithmic}[1]
      \Statex \hspace{-1.2em}\textbf{Initialize:} $j = 1, l=0$, $\hat{\mathbb{Z}}_\mathrm{svm}^1 = \hat{\mathbb{Z}}_\mathrm{cvx}^1 = \hat{\mathbb{Z}}^1$ from \eqref{eq:zhat1} 
      
      \Statex \hspace{-1.2em}\textbf{Inputs:} $\mathbb{W}$, $\epsilon, \beta, N$ and $x_0^j = x_S$ for all $j \in \{1,2,\dots\}$ 
      
      \Statex \hspace{-1.2em}\textbf{Data:} $\tilde{\mathbf{x}}^{1} = [x_0^1, x^1_1, \dots, x_T^1]$, $\hat{\mathcal{P}}_\mathrm{cvx}^2$ formed with \eqref{eq:conv_hull_pol};  
       
      \vspace{1.2mm}
      \WHILE{$j \geq 2$}

      \IF{Points in $\tilde{\mathbf{x}}^{j-1}$ violate \eqref{eq:constraints_nominal}}
      
      \STATE Construct $\hat{\mathbb{Z}}_\mathrm{svm}^j$ with \eqref{eq:zhat_svm}; construct $\hat{\mathcal{X}}_N^j$ with \eqref{eq:term_set_DF};
      
      \ELSE
      
      $\hat{\mathbb{Z}}_\mathrm{svm}^j  = \hat{\mathbb{Z}}_\mathrm{svm}^{j-1}$;
      
      \STATE $l = l+1$; (if $l \geq \frac{\ln 1/\beta}{\ln 1/(1-\epsilon)}$, \textbf{break}; \eqref{eq:math_def_IF} is 
      
      satisfied)
      
      \ENDIF
      
      \IF{$\mathbb{P}([\mathrm{IF}]^{j}) = 0$ desired}
      
      \STATE Construct $\hat{\mathbb{Z}}_\mathrm{cvx}^j$ with \eqref{eq:z_update_conv}; construct $\hat{\mathcal{X}}_N^j$ with \eqref{eq:term_set_DF};

      \IF{Problem \eqref{eq:MPC_R_fin} is feasible with $\hat{\mathbb{Z}}_\mathrm{cvx}^j$}
      
      \STATE Use $\hat{\mathbb{Z}}^j = \hat{\mathbb{Z}}_\mathrm{cvx}^j$ for solving \eqref{eq:MPC_R_fin}; 
      
      \STATE \textbf{break}; ($\mathbb{P}([\mathrm{IF}]^{j}) = 0$ is satisfied)
      
      \ELSE 
      
      \STATE Use $\hat{\mathbb{Z}}^j = \hat{\mathbb{Z}}_\mathrm{svm}^j$ for solving \eqref{eq:MPC_R_fin};
      
      \ENDIF 
      
      \ELSE 
      
      \STATE Use $\hat{\mathbb{Z}}^j = \hat{\mathbb{Z}}_\mathrm{svm}^j$ for solving \eqref{eq:MPC_R_fin};
       
      \ENDIF      
      
      \vspace{0.2cm}     
      
      \STATE Set $\tilde{\mathbf{x}}^j = x_S, t = 0$; 
      
      \WHILE{$0 \leq t \leq T-1$}
     
      \STATE Solve \eqref{eq:MPC_R_fin} and apply MPC control \eqref{eq:inputCL_DF} to \eqref{eq:unc_system};
      
      \STATE Collect states and append $\mathbf{x}^j = [\mathbf{x}^j, x^j_{t+1}]$;  
      
      $t = t+1$
      
      \ENDWHILE
      
      $j = j+1$
      
      \ENDWHILE
      \end{algorithmic}
\end{algorithm}
\section{Numerical Simulations}\label{sec:simul}
We verify the effectiveness of the proposed Algorithm~\ref{alg1} in a simulation example. The source codes are at \href{https://github.com/monimoyb/ConstraintLearning}{\texttt{https://github.com/monimoyb/ConstraintLearning}}.
We find approximate solutions to the following iterative optimal control problem in receding horizon:
\begin{equation}\label{eq:num_opt}
	\begin{array}{clll}
		\hspace{0cm} &V^{j,\star}(x_S)  = \\ [1ex]
		& \displaystyle\min_{u_0^{j},u_1^{j}(\cdot),\ldots}  \displaystyle\sum\limits_{t=0}^{T-1} 10\left \| \bar{x}_t^{j} - x_\mathrm{ref} \right\|^2_2 + 2 \left\| u_t^{j}(\bar{x}^j_t) \right\|^2_2  \\[1ex]
		& ~~~~\text{s.t.,}\\
		& ~~~~~~~~~~~~x_{t+1}^{j} = Ax_t^{j} + Bu_t^{j}(x_t^{j}) + w_t^{j},\\[1.5ex]
		& ~~~~~~~~~~~ \begin{bmatrix} H^\mathrm{b}_x \\ H^\mathrm{ub}_x \end{bmatrix} x_t^{j}
		\leq \begin{bmatrix}20 \times \mathbf{1}_{4 \times 1}\\ 5 \times \mathbf{1}_{2 \times 1} \end{bmatrix},\ \forall w_t^{j} \in \mathbb W,\\
		& ~~~~~~~~~~~ -30 \leq u^j_t(x_t^{j})
		\leq 30,\ \forall w_t^{j} \in \mathbb W,
		\\[.5ex]
		& ~~~~~~~~~~~~ x_0^{j} = x_S,\ t=0,1,\ldots,(T-1),
	\end{array}
\end{equation}
where
\begin{align*}
    \mathbb{W} & = [-0.5,0.5] \times [-0.5,0.5],\\
    A &= \begin{bmatrix} 1 & 1 \\ 0 & 1 \end{bmatrix},~B = [0,1]^\top
\end{align*}
are known. The known and unknown parts of the state constraints are parametrized by the polytopes
$\{x: H_x^\mathrm{b} x \leq 0\}$ and $\{x: H_x^\mathrm{ub} x \leq 0\}$ respectively, where the matrices are given by
\begin{align*}
    H_x^\mathrm{b} = \begin{bmatrix} 1 & 0\\ 0 & 1\\ -1 & 0\\ 0 & -1 \end{bmatrix},~ H_x^\mathrm{ub}  = \begin{bmatrix} 1 & 1 \\ 1 & -1\end{bmatrix}.
\end{align*}
We solve the above optimization problem \eqref{eq:num_opt} with the initial state $x_S = [-15,15]^\top$ and reference point $x_\mathrm{ref} = [5,0]^\top$ for task duration $T=10~\text{steps}$ over all the iterations. Algorithm~\ref{alg1} is implemented with a control horizon of $N=4$, and the feedback gain $K$ in \eqref{eq:term_set_DF} is chosen to be the optimal LQR gain with parameters $Q_\mathrm{LQR}=10{I}_{2 \times 2}$ and $R_\mathrm{LQR} = 2$. The optimization problems are formulated with YALMIP interface \cite{lofberg:05} in MATLAB, and we use the Gurobi solver to solve a quadratic program at every time step for control synthesis. The goal of this section is to show:
\begin{itemize}
    \item Our approach finds an iteration index $\bar{j}$ such that \eqref{eq:math_def_IF} is guaranteed to hold for all iterations $j \geq \bar{j}$. 
    
    \item Performance loss $[\mathrm{PL}]^j$ over Successful Iterations (after $j \geq \bar{j}$) increases as the tolerable probability $\epsilon$ of Iteration Failure is lowered. This highlights the safety vs. performance trade-off.
\end{itemize}

\subsection{Bounding the Probability of Iteration Failure}\label{ssec:resA}
We demonstrate satisfaction of design specification (D1) by Algorithm~\ref{alg1}. First, we focus on the SVM-based approach. We choose an SVM classifier with a Radial Basis kernel function \cite[Chapter~12]{friedman2001elements}. For introducing exploration properties, the SVM classifier $f^0(x)$ was initially warm-started by exciting the system \eqref{eq:unc_system} with random inputs and collecting trajectory data for two trajectories. After that the control process was started by solving \eqref{eq:MPC_R_fin}. The polytopes $\mathcal{P}^{j+1}_\mathrm{svm}$ were generated by taking a convex hull of randomly generated 1000 test points before each iteration, which were classified as $f^j(x^j_\mathrm{test}) \leq 0$ for $j \in \{1,2,\dots\}$.  

We consider two cases of tolerable Iteration Failure, with respective probabilities of $30\%$ and $50\%$, corresponding to $\epsilon = 0.3$ and $\epsilon = 0.5$ (see Table I). The associated estimated constraint sets $\hat{\mathbb{Z}}^{\bar{j}}$ were obtained for $\bar{j} = 5$ and $\bar{j} = 3$ respectively \footnote{We note that the exact value of $\bar{j}$, as well as the associated estimated constraint sets, depend on the disturbance sequence. Running this example several times will yield similar results, but not the exact same results.}.
These sets satisfy design requirement (D1) and are shown in Fig.~\ref{fig:estim_con}. As expected, the constraint set estimated with $\epsilon = 0.5$ is larger than the set estimated with $\epsilon = 0.3$. Both estimated sets partially violate the true constraint set (outlined in black).

Furthermore, in order to verify the certificate obtained using Proposition~\ref{prop:prob_con}, we run $100$ \emph{offline} Monte-Carlo simulations (or \emph{trials}) of iterations by solving \eqref{eq:MPC_R_fin}, with $\hat{\mathbb{Z}}^{1:100} = \hat{\mathbb{Z}}^{\bar{j}}$, for each of the above $\hat{\mathbb{Z}}^{\bar{j}}$ sets, and estimate the actual resulting Iteration Failure probability. This probability is estimated by averaging over 100 Monte Carlo draws of disturbance samples $w_{0:T-1} = [w_0, w_1, \dots, w_{T-1}]$, i.e.,
\begin{align*}
    \mathbb{P}({x}_{0:T} \notin \hat{\mathbb{Z}}_s^{\bar{j}}) \approx \frac{1}{100} \sum_{\tilde{m}=1}^{100} (\mathbf{1}_\mathcal{F}(x_{0:T}))^{\star\tilde{m}},
\end{align*}
where 
$$
(\mathbf{1}_\mathcal{F}(x_{0:T}))^{\star\tilde{m}} = 
\begin{cases}
1,~\textnormal{if } {x_{0:T} \notin \hat{\mathbb{Z}}_s^{\bar{j}} \vert (w_{0:T-1})^{\star \tilde{m}}}, \\
0,~\textnormal{otherwise},\\
\end{cases}
$$
and $(\cdot)^{\star \tilde{m}}$ represents the $\tilde{m}^{\mathrm{th}}$ Monte Carlo sample\footnote{For brevity, with slight abuse of notation, we use $\hat{\mathbb{Z}}_s^{\bar{j}}$ to denote the corresponding state constraints.}. 
\begin{figure}[h]
    \centering
    \includegraphics[width=10cm]{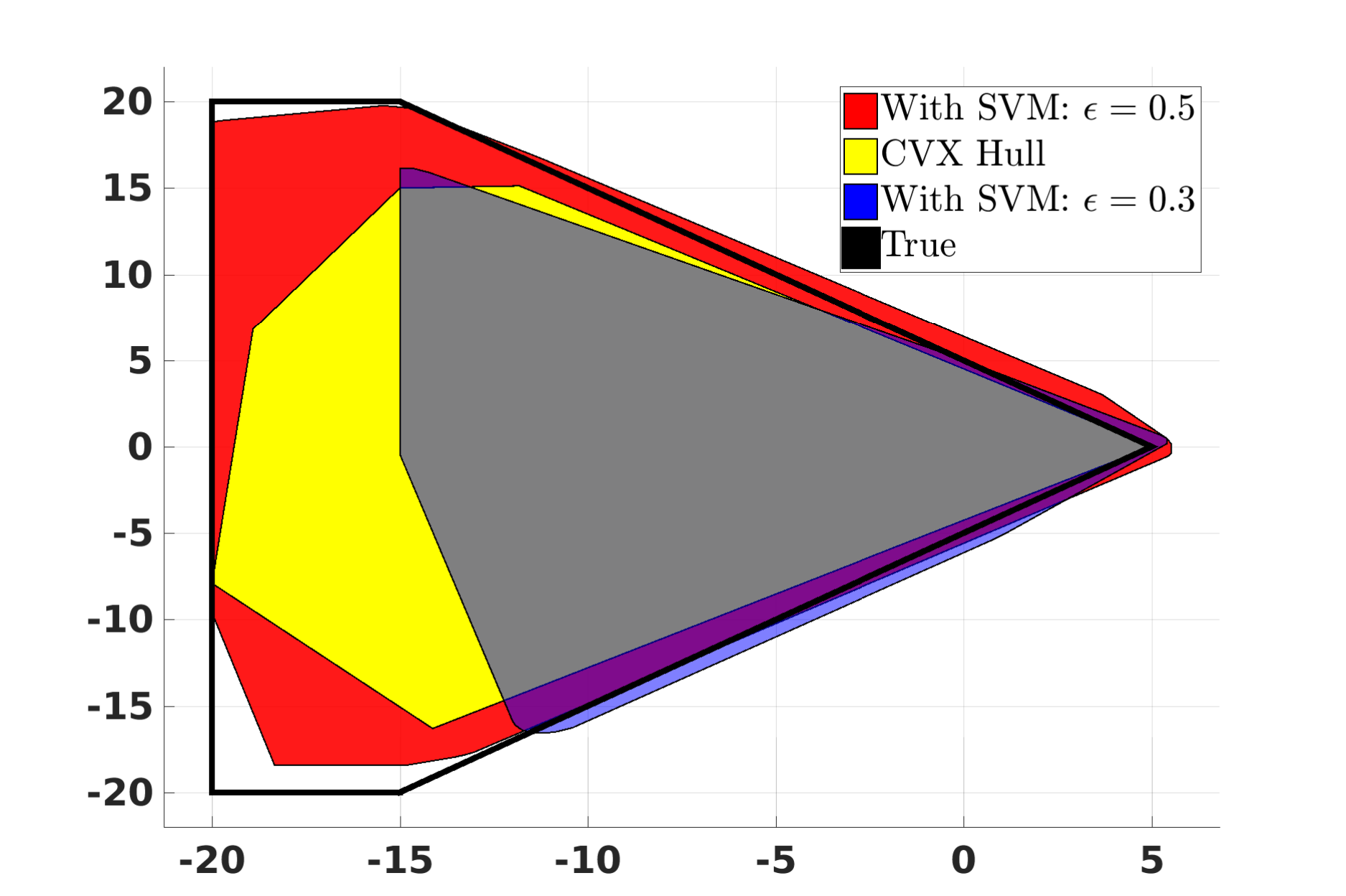}
    \caption{Estimated state constraint sets with varying bounds for $\mathbb{P}([\mathrm{IF}]^j)$.}
    \label{fig:estim_con}
\end{figure}

The values obtained were $\hat{\epsilon} \approx 0.01$ and $\hat{\epsilon} \approx 0.04$ for $\epsilon = 0.3$ and 
$\epsilon = 0.5$ respectively. Thus we see that, in practice, the actual probability of Iteration Failure is about $92\%-96 \%$ lower than the corresponding chosen $\epsilon$. This highlights the conservatism of the bounds given in Proposition~\ref{prop:prob_con}.

We next verify the satisfaction of design requirement (D1) when the estimated constraint sets are obtained using the robust convex hull based approach from Section~\ref{ssec:cvx_sec}. We use the same 100 draws of disturbance sequences $w^{1:100} = [w^1_{0:T-1}, w^2_{0:T-1}, \dots, w^{100}_{0:T-1}]$ as for the SVM-based approach above. The resulting constraint estimate set is shown in Fig.~\ref{fig:estim_con} and is obtained at $\bar{j} = 4$. Using this set in \eqref{eq:MPC_R_fin} ensures no Iteration Failures for all $j \geq \bar{j}$, as proven in Proposition~\ref{prop:zhat_cvx}. These results from Section~\ref{ssec:resA} are summarized in Table~I.

\subsection{Safety vs. Performance Trade-Off}\label{ssec:resB}
For the same Monte Carlo draws of $w^{1:100}$, we approximate the average closed-loop cost $\mathbb{E}[{\hat{\mathcal{V}}^{\bar{j}}}(x_S, w_{0:T})]$ by taking an empirical average over the $100$ Monte Carlo draws,
\begin{align*}
    \mathbb{E}[\hat{\mathcal{V}}^{\bar{j}}(x_S, w_{0:T-1})] \approx \frac{1}{100} \sum_{\tilde{m}=1}^{100} \mathcal{\hat{V}}^{\bar{j}}(x_S, (w_{0:T-1})^{\star \tilde{m}}),
\end{align*}
with $\hat{\mathbb{Z}}^{\bar{j}}$ sets obtained in Fig.~\ref{fig:estim_con}. The cost values are normalized by ${\mathcal{V}}^{\star}(x_S)$, which denotes the empirical average closed-loop cost if $\mathbb{Z}$ had been known. 

The results are summarized in Table I.
We see that the average closed-loop cost shows an inverse relationship with the tolerable Iteration Failure probability $\epsilon$. For lower probabilities of $\mathrm{[IF]}^j$ with $\epsilon = 0.30$, we pay a $3\%$ lower average closed-loop cost compared to ${\mathcal{V}}^{\star}(x_S)$. Allowing for higher probability of $\mathrm{[PL]}^j$ with $\epsilon=0.50$ proves to be cost-efficient, where we pay around $7 \%$ lower average closed-loop cost compared to ${\mathcal{V}}^{\star}(x_S)$. The cost for the approach in Section~\ref{ssec:cvx_sec} is the highest, with a $4 \%$ higher average closed-loop cost compared to ${\mathcal{V}}^{\star}(x_S)$. This directly reflects the safety vs. performance trade-off. \textcolor{black}{This also shows that our approach lowers performance loss for any given $\epsilon$, satisfying (D2) with (D1).} 

\begin{table}[h]
\centering
\caption{}
\label{table:results}
\begin{tabular}{ ||c||c||c||c||} 
 \hline
 $\epsilon$ & $\bar{j}$ & $\hat{\epsilon}$ & $ {\mathbb{E}[\hat{{\mathcal{V}}}^{\bar{j}}(x_S, w_{0:T-1})]}/\mathcal{V}^\star(x_S) \approx$ \\
 \hline
 0 & 4 & 0 & 1.04 \\
 0.3 & 5 & 0.01 & 0.97 \\ 
 0.5 & 3 & 0.04 & 0.93 \\ 
 \hline
\end{tabular}
\end{table}
\vspace{-5pt}

\section{Conclusions}\label{sec: concl}
We propose a framework for an uncertain LTI system to iteratively learn to satisfy unknown polyhedral state constraints in the environment.
From historical trajectory data, we construct an estimate of the true environment constraints before starting an iteration, which the MPC controller robustly satisfies at all times along the iteration. A \emph{safety} certification is then provided for the estimated constraints, if the true (and unknown) environment constraints are also satisfied by the controller in closed-loop. We further highlight a trade-off between safety and controller performance, demonstrating that a controller designed with estimated constraint sets which are deemed highly safe, result in a higher average closed-loop cost incurred across iterations. Finally, we demonstrated the efficacy of the proposed framework via a detailed numerical example. 
\section*{Acknowledgements}
This research was sustained in part by fellowship support from the National Physical Science Consortium and the National Institute of Standards and Technology. The research was also partially funded by Office of Naval Research grant ONR-N00014-18-1-2833.

\renewcommand{\baselinestretch}{.4}
\bibliographystyle{IEEEtran} 
\bibliography{IEEEabrv,bibliography.bib}

\section*{Appendix}
\subsection*{Proof of Proposition~\ref{prop:prob_con}}
Recall matrices $H_x$ and $h_x$ defined in \eqref{eq:constraints_nominal}. Let us denote $\mathbf{H}_x = H_x \otimes {I}_T$ and $\mathbf{h}_x = h_x \otimes {I}_T$. Let $[\mathbf{H}_{x}]_i$ and $[\mathbf{h}_{x}]_i$ denote the $i^\mathrm{th}$ row of $\mathbf{H}_x$ and $\mathbf{h}_x$ respectively. For a fixed initial condition $x^{1:j}_0 = x_S$ and a random disturbance realization $\mathbf{w} = [w_0, w_1, \dots, w_{T-1}]$, consider the corresponding closed-loop trajectory $\mathbf{x}(\mathbf{w}) = [x^\top_0, x^\top_1, \dots, x^\top_T]^\top$. Now consider the following function
\begin{align*}
    Q(\mathbf{w}) :=  \max_i\{ [\mathbf{H}_x]_i\mathbf{x}(\mathbf{w}) - \mathbf{h}_i   \},
\end{align*}
and then define $\hat Q_{N_\mathrm{it}} := \max_{j=1,2,\dots,N_\mathrm{it}}\{ Q(\mathbf{w}^j)\}$, where $\mathbf{w}^{j}$ for $j \in \{1,2,\dots, N_\mathrm{it}\}$ are a collection of independent samples of $\mathbf{w}$ drawn according to $\mathbb{P}$. It follows \cite[Theorem 3.1]{TempoBaiDebbene} that, if $N_\mathrm{it} \geq \frac{\ln1/\beta}{\ln1/(1-\epsilon)}$, then 
\begin{equation*}
    \mathbb{P}^{N_\mathrm{it}} \big[ \mathbb{P}[Q(\mathbf{w}) > \hat{Q}_{N_\mathrm{it}}] \leq \epsilon \big] \geq 1-\beta.
\end{equation*}
Proposition~\ref{prop:prob_con} now follows upon setting $\hat Q_{N_\mathrm{it}}=0$. 
\end{document}